\documentclass{article}
\usepackage[utf8]{inputenc}

\usepackage[letterpaper,margin=1in]{geometry}
\usepackage[OT1]{fontenc}

\usepackage{amsmath,amssymb,amsthm,amsfonts,latexsym,bbm,xspace,thm-restate}
\usepackage{graphicx,float,mathtools,braket,slantsc}
\usepackage{enumitem,booktabs,forest,mathdots,soul}
\usepackage[useregional]{datetime2}
\DTMusemodule{english}{en-US}

\usepackage{qcircuit,tikz}

\usepackage[colorlinks,citecolor=blue,,linkcolor=magenta,bookmarks=true]{hyperref}
\usepackage[capitalise]{cleveref}

\usepackage{multicol,array}
\usepackage{lipsum,framed}
\usepackage{caption} 

\newcommand{\remove}[1]{}

\newcommand{\changed}[1]{{\textcolor{purple} #1}}

\newtheorem{theorem}{Theorem}[section]
\newtheorem*{theorem*}{Theorem}
\newtheorem{lemma}[theorem]{Lemma}

\newtheorem{proposition}[theorem]{Proposition}
\newtheorem{corollary}[theorem]{Corollary}

\crefname{result}{Result}{Results}

\newtheorem{fact}[theorem]{Fact}
\crefname{fact}{Fact}{Facts}

\theoremstyle{definition}
\newtheorem{definition}[theorem]{Definition}

\newtheorem{remark}[theorem]{Remark}

\newtheorem{problem}[theorem]{Problem}

\newcommand{\ms}[1]{\ensuremath{\mathsf{#1}}}

\let\Pr\relax
\DeclareMathOperator{\Pr}{Pr}
\DeclareMathOperator*{\E}{\mathbb{E}}
\DeclareMathOperator*{\Ex}{\E}

\DeclareMathOperator{\tr}{tr}

\newcommand{\trdist}[1]{\norm{#1}_{\tr}}

\DeclareMathOperator{\poly}{poly}

\newcommand{\ketbra}[2]{\ket{#1}\!\bra{#2}}

\newcommand{\e}{\varepsilon}
\newcommand{\eps}{\varepsilon}

\newcommand{\class}[1]{\ensuremath{\mathsf{#1}}\xspace}
\mathchardef\mhyphen="2D 

\newcommand{\NP}{\class{NP}}
\newcommand{\UP}{\class{UP}}

\newcommand{\AM}{\class{AM}}
\newcommand{\coAM}{\class{coAM}}

\newcommand{\QMA}{\class{QMA}}

\newcommand{\calD}{\mathcal{D}}

\newcommand{\calL}{\mathcal{L}}
\newcommand{\calM}{\mathcal{M}}

\newcommand{\calP}{\mathcal{P}}

\newcommand{\calR}{\mathcal{R}}

\newcommand{\calV}{\mathcal{V}}

\newcommand{\xv}{\mathbf{x}}
\newcommand{\wv}{\mathbf{w}}

\DeclarePairedDelimiter\absd{\lvert}{\rvert}
\DeclarePairedDelimiter\normd{\lVert}{\rVert}

\DeclarePairedDelimiter\floord{\lfloor}{\rfloor}
\DeclarePairedDelimiter\ceild{\lceil}{\rceil}

\newcommand{\abs}[1]{\absd*{#1}}
\newcommand{\norm}[1]{\normd*{#1}}

\newcommand{\floor}[1]{\floord*{#1}}
\newcommand{\ceil}[1]{\ceild*{#1}}

\newcommand{\etal}{et al.\xspace}

\newcommand{\deq}{\coloneq}
\newcommand{\getsr}{%
  \mathrel{\vbox{\offinterlineskip\ialign{%
    \hfil##\hfil\cr
    $\scriptstyle\normalfont\textsc{r}$\cr
    \noalign{\kern0.02ex}
    $\leftarrow$\cr
}}}}
\newcommand{\concat}{\circ}

\renewcommand{\set}[1]{\left\{#1\right\}}
\newcommand{\zo}{\set{0,1}}

\newcommand{\XOR}{\textsc{xor}}

\newcommand{\QSZK}{\ms{QSZK}}
\newcommand{\hvQSZK}{\ms{hvQSZK}}
\newcommand{\SZK}{\ms{SZK}}
\newcommand{\NISZK}{\ms{NISZK}}
\newcommand{\QSWI}{\ms{QSWI}}
\newcommand{\SWI}{\ms{SWI}}
\newcommand{\hvSWI}{\ms{hvSWI}}
\newcommand{\hvQSWI}{\ms{hvQSWI}}
\newcommand{\pubQSWI}{\ms{pubQSWI}}

\newcommand{\seps}{\eps_{\ms{S}}}
\newcommand{\ceps}{\eps_{\ms{C}}}
\newcommand{\wieps}{\eps_{\ms{WI}}}

\newcommand{\R}{\calR}
\newcommand{\N}{\mathbb{N}}
\newcommand{\batch}[2]{{#1}^{\otimes #2}}
\newcommand{\LL}{\ms{L}}
\newcommand{\View}{\ms{View}}
\newcommand{\negl}{\ms{negl}}

\newcommand{\newprotocol}[5]{
  \begin{minipage}{\linewidth}
    \medskip
    \begin{framed}
    \small
    {\bf  #1}\\
    {\bf Inputs:} #2 
    \medskip
    \begin{enumerate}
        #3 
    \end{enumerate}
  \end{framed}
  \vspace{-1.5em}
  \captionof{figure}{#4.} 
  \label{#5}
  \medskip
  \end{minipage}
}

\newcommand{\figbox}[4]{
  \begin{minipage}{\linewidth}
    \medskip
    \begin{framed}
    \small
    {\bf  #1}
    \medskip
    \\
    #2 
  \end{framed}
  \vspace{-1.5em}
  \captionof{figure}{#3.} 
  \label{#4}
  \medskip
  \end{minipage}
}

\newcommand{\reg}[1]{\mathsf{#1}}

\title{Quantum Statistical Witness Indistinguishability}
\author{
    Shafik Nassar\\\small{\textsl{The University of Texas at Austin}}\\\small{\texttt{\href{mailto:shafik@cs.utexas.edu}{shafik@cs.utexas.edu}}}
    \and 
    Ronak Ramachandran\\\small{\textsl{The University of Texas at Austin}}\\\small{\texttt{\href{mailto:ronakr@utexas.edu}{ronakr@utexas.edu}}}
}
\hypersetup{
    pdftitle={Quantum Statistical Witness Indistinguishability},
    pdfauthor={Shafik Nassar, Ronak Ramachandran},
}
\date{September 2025}

\begin{document}

\maketitle

\begin{abstract}
    Statistical witness indistinguishability is a relaxation of statistical zero-knowledge which guarantees that
    the transcript of an interactive proof reveals no information about \emph{which} valid witness the prover used to generate it.
    In this paper we define and initiate the study of $\mathsf{QSWI}$, the class of problems with \emph{quantum} statistically witness indistinguishable proofs. 
    
    Using inherently quantum techniques from Kobayashi (TCC 2008), we prove that any problem with an \emph{honest-verifier} quantum statistically witness indistinguishable proof has a $3$-message public-coin \emph{malicious-verifier} quantum statistically witness indistinguishable proof.
    There is no known analogue of this result for \emph{classical} statistical witness indistinguishability.
    As a corollary, our result implies $\mathsf{SWI} \subseteq \mathsf{QSWI}$.
    
    Additionally, we extend the work of Bitansky et al. (STOC 2023) to show that \emph{quantum} batch proofs imply \emph{quantum} statistically witness indistinguishable proofs with inverse-polynomial witness indistinguishability error.
\end{abstract}

\section{Introduction}
\emph{Witness indistinguishability} (WI) \cite{FS90} is
a property of interactive proofs \cite{GMR85}
that is a
natural relaxation of
\emph{zero-knowledge}.
A WI proof for an 
$\NP$ relation $\calR$
is a proof in which the verifier does not learn which valid witness the honest prover used in the interaction. That is, if $(x,w),(x,w') \in \calR$, then the view of the interaction transcript when the prover uses $w$ is indistinguishable from the transcript when the prover uses $w'$.\footnote{Observe that for this definition to make sense, the prover has to be computationally bounded when given a witness, otherwise it can just find a canonical witness regardless of the witness it is given.}

Recently, the complexity class $\SWI$ was introduced by 
Bitansky~\etal~\cite{BKPRV24} in the context of studying \emph{batch proofs for $\NP$}\footnote{We explain the notion of a batch proof later on.}.
This class consists of all problems in $\NP$ that have a \emph{statistical} WI proof.
This class can be viewed as a relaxation of the more famous class, $\SZK$ \cite{Vad99}, consisting of all problems that have a statistical zero-knowledge proof.\footnote{Note, however, that $\SZK$ is not contained in $\SWI$ due to the technical fact that $\SWI$ is only defined for $\NP$ relations, whereas $\SZK$ does not have any such inherent constraint. Despite that, a result by Nguyen and Vadhan \cite{NV06} implies that $\NP \cap \SZK \in \SWI$.}

Bitansky~\etal conjecture that $\NP \not\subseteq\SWI$.
While this seems like a plausible conjecture, there is no evidence to support it.
In fact, almost nothing is known about the complexity class $\SWI$.
This is in stark contrast to the success story of studying $\SZK$:
the class $\SZK$ has a complete problem, is closed under complement, is robust to whether the verifier is malicious or honest, and is also robust to whether the verifier is allowed to keep secrets or not (private-coin vs. public-coin).
This allows us to derive meaningful conditional results about $\SZK$.
For example, we have evidence that $\NP \not\subseteq\SZK$: 
we know that $\SZK$ has public-coin, constant-round proofs \cite{For87} and that it is closed under complement \cite{Oka96}.
In other words, $\SZK\subseteq \AM \cap \coAM$.
Therefore, $\NP \not\subseteq\SZK$ unless $\NP \subseteq\coAM$
and the polynomial hierarchy collapses \cite{BHZ87}.
Unfortunately, none of those things are known for $\SWI$.

In this work, we show that the state of affairs is very different in a \emph{quantum world}.
We present the complexity class $\QSWI$, which consists of languages that have a statistical WI \emph{quantum} interactive proof.
Similar to $\QSZK$ \cite{Wat02}, we study a model in which both the prover and verifier are quantum machines, and they are allowed to exchange quantum information.
Statistical indistinguishability, in this case, is defined in terms of trace distance between quantum states.
To the best of our knowledge, this is the first time quantum statistical WI proofs have been considered in the literature.

\subsection{Our Contributions}
\paragraph{Fundamental results about $\QSWI$.}
We start by formally defining $\QSWI$ and its variants, such as $\hvQSWI$ where the statistical WI property is guaranteed only against honest verifiers, and $\pubQSWI$ in which the verifier only sends random (classical) bits.
By definition: $\pubQSWI \subseteq \QSWI \subseteq \hvQSWI$.
We show that in fact:
\begin{theorem}\label{thm:main1}
    Any problem in $\hvQSWI$ has a $3$-message, public-coin quantum interactive proof that satisfies quantum statistical witness indistinguishability against malicious verifiers.
    In particular: $\pubQSWI = \QSWI = \hvQSWI$.
    Moreover, the witness indistinguishability error in the resulting protocol is polynomially related to that in the original protocol.
\end{theorem}
\noindent
As a corollary of \cref{thm:main1} and the fact $\SWI \subseteq \hvSWI \subseteq \hvQSWI$:
\begin{corollary}\label{cor:main2}
    $\SWI \subseteq \QSWI$.
\end{corollary}
We remark that despite $\QSWI$ being an extension of the classical $\SWI$,
\cref{cor:main2} is not immediate without \cref{thm:main1},
since classical witness indistinguishability is only guaranteed against efficient classical, malicious verifiers.

\paragraph{Connection to quantum batch proofs.}
Next, we formally define quantum batch proofs.
Roughly speaking, a batch proof allows a prover to convince a verifier that many statements are all true, without communicating all of the witnesses.
We stress that there is no privacy guarantee (like ZK or WI) for the prover in this setting.
A $\rho$-compressing quantum batch proof is a quantum interactive proof for proving the correctness of $t$ statements, in which the prover and the verifier exchange less than $O(\rho t)$ qubits.

With this notion in hand, we prove that the work of Bitansky~\etal~\cite{BKPRV24} -- compiling batch proofs to witness indistinguishability proofs -- actually extends to the quantum setting.
That is, we prove the following theorem:
\begin{theorem}\label{thm:main2}
    Let $\calR$ be any $\NP$ relation.
    If $\calR$ has a $\rho$-compressing quantum batch proof, then $\calR$ has a quantum interactive proof with a non-uniform honest prover that satisfies quantum statistical witness indistinguishability against honest verifiers, with witness indistinguishability error $\sqrt{\rho}$.
\end{theorem}
First, we emphasize that the batch proof we start from only satisfies completeness and soundness, and has no witness indistinguishability (or any other privacy) guarantee.
Second, we remark that the resulting QSWI protocol satisfies only a weak notion of witness indistinguishability, since the WI error is $\sqrt{\rho}$, which is inverse polynomial, as opposed to the standard negligible error.
Additionally, the honest prover is crucially non-uniform.
These weaknesses are inherited by \cite{BKPRV24}.
Second, using the ``moreover'' part of \cref{thm:main1}, we conclude that regardless of the batch proof we start from, the resulting QSWI protocol can be made public-coin, $3$-message, and secure against malicious verifiers.\footnote{The transformation preserves the number of messages, therefore this remark is only useful if the batch proof we start from is interactive. If the batch proof is non-interactive, then the resulting QSWI proof is also non-interactive and therefore trivially ``public-coin'' and secure against malicious verifiers, since the verifier does not communicate anything to the prover.}
We note that classical batch proofs are \emph{not} known to imply classical \emph{malicious}-verifier SWI proofs.

\subsection{Discussion}
\paragraph{Understanding statistical witness indistinguishability.}
One takeaway from our results is that $\QSWI$ is easier to reason about than $\SWI$: we prove facts about $\QSWI$ that we do not know about $\SWI$.
At the same time, we prove that $\SWI \subseteq \QSWI$.
This means that we can potentially prove bounds on $\SWI$ by proving bounds on $\QSWI$.
For example, \emph{if} we prove that $\QSWI$ is closed under complement, then $\SWI$ is not equal to $\NP$ (unless the polynomial hierarchy collapses).

\paragraph{On the power of $\QSWI$.}
While $\SWI$ is conjectured to not contain $\NP$,
we stress (again) that there is no evidence to support this conjecture.
Since $\QSWI$ is at least as powerful as $\SWI$, we refrain from conjecturing that $\NP \not\subseteq \QSWI$.
We believe there is value in trying to tackle the problem ``from both directions''.

As such, readers that believe $\NP \not\subseteq \QSWI$ might interpret \cref{thm:main2} as evidence that quantum batching of $\NP$ proofs is unlikely.
On the other hand, the more ``optimistic'' reader might interpret \cref{thm:main2}
as an approach to place $\NP$ in $\QSWI$.
We reiterate that \cref{thm:main2} can only give a weak notion of witness indistinguishability, since the witness indistinguishability error using this approach is inherently inverse polynomial rather than negligible.
But since no unconditional ZK/WI proof for all of $\NP$ is known, even a weak QSWI proof for $\NP$ would be a very interesting result!

\paragraph{Quantum batch proofs.}
We think that quantum batching of proofs is an interesting subject in its own right.
In the classical setting, a line of work by Reingold, Rothblum and Rothblum showed how to batch $\UP$ proofs \cite{RRR16,RRR18}.
Another line of work showed how to batch problems in the class $\NISZK$~\cite{KRRSV20,KRV21,MNRV24} even with efficient provers~\cite{KRV24}.
Extending those results using quantum interactive proofs is an intriguing open problem: can we batch classes beyond $\NISZK$ (maybe $\SZK$) or $\UP$ (maybe $\NP$) using quantum interactive proofs?
Notably, under cryptographic assumptions, one can push batching to all of $\NP$~\cite{CJJ21a,CJJ21b,WW22,CGJJZ23}, and even support monotone batching policies \cite{BBKLP24,NWW24,NWW25}.

\paragraph{Definitional choice.}
We defined witness indistinguishability with respect to a classical relation, and thus, classical witnesses are used in the definition.
One can imagine defining quantum witness indistinguishability with respect to quantum witnesses, which leads to a ``$\QMA$ version'' of $\QSWI$.
We defer the study of such a class to future work, although it seems like our results should seamlessly extend to the quantum witness setting.

Finally, we point out the work of Popa~\cite{Pop11} which studies witness indistinguishability of \emph{classical} interactive proofs against quantum (malicious) verifiers.
In contrast, our setting studies \emph{quantum} interactive proofs.

\subsection{Technical Overview}

\paragraph{The problem with $\SWI$.}
Before we explain our techniques, we discuss the challenge in studying the classical $\SWI$.
In the classical setting, most of the success in studying $\SZK$ came from proving the existence of a complete problem \cite{SV03}, and using the zero-knowledge protocol for that problem as well as its structure to derive results about all of $\SZK$.
Establishing a complete problem for $\SWI$ seems much more elusive.

Work prior to \cite{SV03} managed to prove some results about $\SZK$, such as closure under complement and having public-coin protocols \cite{Oka96}.
Let us discuss the techniques used by Okamoto \cite{Oka96} and why they fail in the witness indistinguishability setting:
\begin{itemize}
    \item To prove that $\SZK$ has public-coin protocols, Okamoto uses a set lower-bound protocol \cite{GS86}.
    This results in an inefficient prover. Achieving witness indistinguishability with an inefficient prover is trivial since the prover can always find a ``canonical'' witness and send it to the verifier.
    \item To prove closure under complement, Okamoto uses a technique developed by Fortnow \cite{For87} which dictates the verifier runs the zero-knowledge simulator.
    While witness indistinguishability can be defined using a simulator, that simulator is inherently inefficient (unless WI is equivalent to some sort of ZK). This results in an inefficient verifier.
\end{itemize}
Therefore,  when porting ``ZK'' techniques to the ``WI'' setting, one should ensure that any transformation preserves the prover's efficiency, and does not require any party to run the simulator.

\paragraph{Our approach for $\QSWI = \hvQSWI$.}
We start by recalling Kobayashi's approach \cite{Kob08} to proving that $\QSZK = \hvQSZK$:
\begin{enumerate}
    \item Using the round compression technique of \cite{KW00}, we can compress any honest-verifier QSWI proof to only $3$ messages.
    Kobayashi \cite{Kob08} proved that this transformation preserves honest-verifier QSZK.

    \item Marriott and Watrous \cite{MW05} showed
    how to compile a $3$-message private-coin QIP into a $3$-message public-coin QIP, in which the verifier sends a single classical bit.
    Again, Kobayashi \cite{Kob08} proved that this transformation preserves honest-verifier QSZK.

    \item Finally, using Watrous' quantum rewinding technique \cite{Wat06}, Kobayashi \cite{Kob08} proved that the resulting public-coin protocol (where the verifier sends only a single classical bit) is also QSZK against malicious verifiers.
\end{enumerate}

Next, we observe that $\QSWI$ can also be defined using a simulator: a QIP  $(P,V)$ with an efficient prover is QSWI if and only if there exists a (potentially unbounded) simulator that simulates the verifier's view.
By carefully analyzing all of the previous transformations,
we prove that they all preserve the honest prover's efficiency and they still work with unbounded simulators.
This is in contrast to Okamoto's techniques \cite{Oka96} for classical SZK.
As a result, we get the following:
\begin{enumerate}
    \item Round compression for honest-verifier QSWI protocols, where the resulting protocol has only $3$ messages.

    \item Private-coin to public-coin transformation for $3$-message honest-verifier QSWI.

    \item Security against malicious-verifiers for $3$-message honest-verifier QSWI proofs.

\end{enumerate}

We point out that while Kobayashi also proved that completeness errors for QSZK can be generically eliminated, we could not use his protocol to prove the same for QSWI, because it may not preserve the prover's efficiency.
However, the above three transformations do preserve perfect completeness.

\paragraph{Our approach for quantum batching to QSWI.}
The idea behind the transformation of \cite{BKPRV24} is that batch proofs necessarily lose information about some of the witnesses, 
since the number of total communicated bits is less than the lengths of all of the witnesses.

We now explain how to utilize this idea.
Let $(P,V)$ be a $\rho$-compressing batch proof for the relation $\calR$ with batches of size $t$.
For simplicity, assume it is non-interactive ($P$ just sends a proof of size $O(\rho t)$, and $V$ either accepts or rejects).
The QSWI prover
samples $t-1$ ``dummy'' instance-witness pairs, and proves their correctness, alongside the target instance, using the batch proof.
The hope is that by choosing the correct distribution, the witness of the target instance gets lost.

In more detail, 
on input $(x,w) \in \calR$, the QSWI prover sets up a batch as follows:
\begin{enumerate}
    \item Chooses a random index $j \getsr [t]$.
    \item Sets $x_j = x$ and $w_j = w$.
    \item For all $i \neq j$, the prover chooses $(x_i, w_i^0, w_i^1)$ from some distribution $D$ that satisfies $(x_i,w_i^b) \in \calR$.
    \item Chooses a random string $b_1\cdots b_t \getsr \zo^{t}$, and uses witness $w_i^{b_i}$ for all $i \neq j$.
\end{enumerate}
The QSWI prover runs $P$ to obtain a proof $\pi$ and sends $\pi$ as well as $x_1,\ldots,x_y$ to the verifier.
The verifier simply checks that $x_j = x$ for some $j \in [t]$, then runs $V$ and accepts or rejects accordingly.
Completeness and soundness of this generic protocol follows by the completeness and soundness of $(P,V)$, and the fact that
$D$ is a distribution over correct instances and their corresponding witnesses, which guarantees that
$(x_i,w_i^{b_i}) \in \calR$ for all $i \neq j$.

Proving WI is trickier.
Let $x,w_0,w_1$ such that $(x,w_0),(x,w_1) \in \calR$.
Imagine that $D$ is simply fixed to be $(x,w_0,w_1)$.
Now define a function $f_{x,w_0,w_1}(b_1\cdots b_t) = P\big(x^t, w^{b_1},\cdots,w^{b_t}\big)$, i.e., runs the batch proof prover $P$ on a batch consisting of $t$ copies of $x$, and each input bit $b_i$ indicates which of the witnesses $w_0,w_1$ to use for copy $i$.
Observe that $f$ is a $\rho$-compressing function, and it cannot distinguish what $j$ is (where the ``real'' $x$ was planted), therefore we expect it to lose information on the bit $b_j$.
This exactly corresponds to the proof losing information on which witness the prover used for the real instance.

It turns out that this intuition translates well to the quantum setting, where $f$ maps $t$ classical bits to $\rho t$ qubits.
This is captured by the quantum distributional stability lemma of \cite{Dru12}, which we adapt in \cref{sec:q-info} to our setting.

The problem with the previous approach is that the prover does not get both witnesses $w_0,w_1$ in the WI game, therefore it cannot sample from the distribution $D$ that we defined.
Bitansky~\etal~\cite{BKPRV24} handle this problem using a sparse minimax theorem \cite{LY94}, at the price of using a non-uniform honest prover.
For more details, see \cref{sec:batch-to-qswi}.

\section{Preliminaries}
For any positive integer $n \in \N$, we denote $[n] = \set{1,\ldots,n}$.
For any random variables $X, Y$ and positive integer $t \in \N$, we denote by $X^{\otimes n}$ the distribution of $t$ independent samples of $X$, and denote by $X \concat Y$ the distribution of concatenating a sample of $X$ with a samples of $Y$.
We denote by $U$ the uniform distribution over $\zo$.
For the uniform distribution specifically, we denote $U^t \deq U^{\otimes t}$, i.e., the uniform distribution over $\zo^t$.
For any $j \in [t], \beta \in \zo$, we denote by $U^t|_{j \gets \beta}$ the uniform distribution conditioned on bit $j$ being equal to $\beta$.
For any set $S$, we write $x \getsr S$ to denote a random variable $x$ that is sampled uniformly from $S$.
For any $\NP$ relation $\calR$ and any $x$, we denote $\calR(x) = \set{ w : (x,w) \in \calR}$.

\subsection{Quantum Information}\label{sec:q-info}
The following definition is adapted from \cite{Dru12}, and is actually just a specific case of \cite[Definition 8.9]{Dru12} that suffices for our use-case.
\begin{definition}[Quantum Distributional Stability \cite{Dru12}]\label{def:qds}
    Let $t,t' \in \N$.
    Given any mapping $f$ that inputs $t$ classical bits and outputs $t'$ qubits, we define for each $j \in [t]$
    \[ \gamma_j \deq \Ex_{\beta \sim U} \left[ \trdist{f\big(U^t|_{j \gets \beta}\big) - f\big(U^t\big)} \right]. \]
    For any $\delta \in [0,1]$, we say that $f$ is $\delta$-quantumly-distributionally stable ($\delta$-QDS) if
    \[ \E_{j \getsr [t]} [\gamma_j] \leq \delta. \]
\end{definition}

\begin{lemma}[QDS of Compressing Functions \cite{Dru12}]\label{lem:qds}
    Let $t,t' \in \N$, and let $f$ be any mapping that inputs $t$ classical bits and outputs at most $t'$ qubits.
    Then $f$ is $O(\sqrt{t'/t})$-QDS.
\end{lemma}

As an immediate corollary of \cref{lem:qds}, we get a quantum version of \cite[Lemma 8]{Del16}:
\begin{corollary}
    \label{lem:compress}
    Let $t \in \N$ and $\rho \in (0,1)$, and let $f$ be any mapping that inputs $t$ classical bits and outputs at most $\rho t$ qubits.
    Then
    \[ \E_{j \getsr [t]} \left[ \trdist{f\big(U^t|_{j \gets 0}\big) - f\big(U^t|_{j \gets 1}\big)} \right] = O(\sqrt{\rho}). \]
\end{corollary}
\begin{proof}
Let $f$ be a function as described in the premise. By \cref{lem:qds}, we have that $f$ is $O(\sqrt{t'/t})$-QDS.
    \begin{align*}
        \E_{j \getsr [t]} \left[ \trdist{f\big(U^t|_{j \gets 0}\big) - f\big(U^t|_{j \gets 1}\big)} \right] &=
        \E_{j \getsr [t]} \left[ \trdist{f\big(U^t|_{j \gets 0}\big) - f\big(U^t\big) + f\big(U^t\big) - f\big(U^t|_{j \gets 1}\big)} \right] \\
        &\leq \E_{j \getsr [t]} \left[ \trdist{f\big(U^t|_{j \gets 0}\big) - f\big(U^t\big)} \right] + \E_{j \getsr [t]} \left[ \trdist{f\big(U^t|_{j \gets 1}\big) - f\big(U^t\big)} \right] \\
        &= 2 \E_{\beta \sim U} \left[ \E_{j \getsr [t]} \left[ \trdist{f\big(U^t|_{j \gets \beta}\big) - f\big(U^t\big)} \right]\right] \\
        &= 2 \E_{\changed{j \getsr [t]}} \left[ \E_{\changed{\beta \sim U}} \left[ \trdist{f\big(U^t|_{j \gets \beta}\big) - f\big(U^t\big)} \right]\right] \\
        &= 2 \E_{j \getsr [t]}[\gamma_j] \\
        &= O(\sqrt{t'/t})
    \end{align*}
    and the claim follows.
\end{proof}

\subsection{Quantum Interactive Proofs}
In this section, we define quantum interactive proofs for relations, rather than languages. The difference is that the (honest) prover gets an additional input which can be viewed as the witness for the instance.
This modification is crucial when studying proofs with efficient-prover, especially when we are interested in ``privacy'' properties like zero-knowledge or witness indistinguishability.

\paragraph{Quantum Interactive Proofs.}
A $k(n)$-round quantum interactive proof for a relation $\calR$ is a pair of mappings $P$ and $V$.
For any $x \in \zo^n$ and $w \in \zo^m$, we interpret $V(x)$ (respectively, $P(x,w)$) as $k\deq k(n)$ quantum circuits $(V_1(x),\ldots,V_k(x))$ (respectively, $(P_1(x,w),\ldots,P_k(x,w))$).
For each $j \in [k]$, the qubits upon which each $V_j(x)$ (respectively, $P_j(x,w)$) acts are divided into two sets: $q_{V}(n)$(respectively,  $q_{P}(n)$) of them are private qubits and $q_{M}(n)$ are (public) message qubits. Note that $V_j(x)$ and $P_j(x,w)$ have the same number of message qubits, and that $q_{V}(n)$, $q_{P}(n)$ and $q_{M}(n)$ are all fixed across all $j \in [k]$.
There is a designated private verifier qubit for the circuit $V_k(x)$ which is the output qubit, which indicates whether the verifier accepts or rejects.
Given a prover $P$, a verifier $V$ and an input $(x,w)$, we define a quantum circuit $\langle P(w),V\rangle(x)$ acting on $q_{V}(n) + q_{M}(n) + q_{P}(n)$ as follows: apply the following circuits 
\[ P_1(x,w), V_1(x),\ldots, P_k(x,w), V_k(x) \]
in order, each to the prover/message qubits or to the verifier/message qubits, accordingly.
The output of $\langle P(w),V\rangle(x)$ is the output qubit when measured in the standard basis.
A QIP should satisfy the following properties:
\begin{enumerate}
    \item \textbf{Verifier Efficiency:} There exists a polynomial $\poly(n)$ such that for any $x \in \zo^n$, the mapping $V(x)$ can be computed in $\poly(\abs{x})$ time. Specifically, each circuit $V_j(x)$ can be computed in polynomial time and
    $k(n)$, $q_{V}(n)$ and $q_{M}(n)$ are all bounded by $\poly(n)$.

    \item \textbf{Completeness:} If $(x,w) \in \calR$ then $\Pr[\langle P(w),V\rangle(x)] \geq 1- \ceps$.

    \item \textbf{Soundness:} If $x \notin \LL(\calR)$ then for any (potentially unbounded) $P^*$ we have $\Pr[\langle P,V\rangle(x)] \leq \seps$. Note that we do not need to pass any input to $P$ since it can just be hard-coded.
\end{enumerate}
Additionally, we say that a QIP has an \textbf{efficient prover} if it satisfies the following additional property:
\begin{enumerate}[start=4]
    \item \textbf{Prover Efficiency:} There exists a polynomial $\poly(n)$ such that for any $(x,w) \in \R$ , the mapping $P(x,w)$ can be computed in $\poly(\abs{x} + \abs{w})$ time. Specifically, each circuit $P_j(x,w)$ can be computed in polynomial time and
    $q_{P}(n)$ is also bounded by $\poly(\abs{x} + \abs{w})$. 
\end{enumerate}

\paragraph{The View of a Verifier.}
The view of the verifier $V$ at round $j$ is defined as the reduced state of the verifier and message qubits after the prover sent its $j^{\text{th}}$ message. That is, for any $x \in \zo^n$ and any $w \in \zo^m$, we denote by $\View_V(x,w,j)$ the reduced state of the verifier and message qubits after $P_j(x,w)$ was applied.
Note that given $\View_V(x,w,j)$, we can always compute the ``view'' of the verifier after it sends its $j^{\text{th}}$ message simply by applying $V_j(x)$. Furthermore, by the efficiency of $V$, this can be computed in polynomial time.

\section{Quantum Statistical Witness Indistinguishability}
We can now present our definition of quantum statistical witness indistinguishability.
Throughout this entire section, $m$ will denote the number of messages in the QIP, while $k$ will denote the number of rounds.
In general, $k = \ceil{m/2}$.
\begin{definition}[Honest-Verifier Quantum Statistical Witness Indistinguishability]\label{def:hvQSWI}
    Let $m \in \N$ and $k = \ceil{m/2}$.
    A $m$-message QIP $(P,V)$ for an $\NP$ relation $\R$ is an honest-verifier QSWI proof with $\wieps(n)$ witness-indistinguishability error if it has an efficient prover and for any $x,w_0,w_1$ such that $(x,w_0),(x,w_1) \in \R$, and any $j \in [k]$ we have
    \[ \trdist{\View_{V}(x,w_0,j) - \View_{V}(x,w_1,j)} \leq \wieps(\abs{x}). \]

    We denote by $\hvQSWI\big(m,1-\ceps, \seps, \wieps\big) $ the class of all $\NP$ languages for which there exists an $\NP$ relation with an honest-verifier QSWI proof
    with $m$ messages, $\ceps$ completeness error, $\seps$ soundness error, and $\wieps$ WI error.
    We abuse notation and say that a QIP $(P,V)$ is a $\hvQSWI\big(m,1-\ceps, \seps, \wieps\big)$ protocol if it is an $m$-message honest-verifier QSWI protocol with the appropriate errors.
    Furthermore, we denote $\hvQSWI \deq \hvQSWI\big(\poly(n),2/3, 1/3, \negl(n)\big)$ where $\poly(n)$ is some polynomial and $\negl(n)$ is some negligible function.
\end{definition}

\begin{remark}
    Observe that QSWI becomes trivial if we allow the honest prover $P$ to run in exponential time, since it can just ignore its input witness $w$ and find some canonical $w'$ (say, the first lexicographical valid witness for $x$) and just send that in the clear. In this case, the view of the verifier does not depend on $w$.
\end{remark}

\begin{definition}[Malicious-Verifier Quantum Statistical Witness Indistinguishability]\label{def:mvQSWI}
    An honest-verifier QSWI proof $(P,V)$ for an $\NP$ relation $\R$ is also \emph{malicious}-verifier QSWI proof with $\wieps(n)$ WI error if for any efficient verifier $V^*$ and any $x,w_0,w_1$ such that $(x,w_0),(x,w_1) \in \R$, and any $j \in [k]$ we have
    \[ \trdist{\View_{V^*}(x,w_0,j) - \View_{V^*}(x,w_1,j)} \leq \wieps(\abs{x}). \]

    We define $\QSWI\big(m,1-\ceps, \seps, \wieps\big)$ and $\QSWI$ completely analogously to the honest-verifier case.
\end{definition}

\begin{fact}[WI and ZK with unbounded simulation]\label{fact:WIZK}
    A QIP $(P,V)$ is
    an $\hvQSWI\big(m,1-\ceps, \seps, \wieps\big)$ protocol if and only if there exists a (potentially unbounded) simulator $S$ such that for any $(x,w) \in \R$ and any $j \in [k]$:
    \[ \trdist{\View_{V}(x,w,j) - S(x,j)} \leq O(\wieps). \]
\end{fact}
\begin{proof}
    Let $(P,V)$ be an $\hvQSWI\big(m,1-\ceps, \seps, \wieps\big)$ protocol.
    We define the unbounded simulator $S(x,j)$ as follows:
    \begin{enumerate}
        \item Find a canonical witness (say, the lexicographically first witness) $w^*$ such that $(x,w^*) \in \calR$.
        \item Simulate the interaction $\langle P(w),V\rangle(x)$, and output $\View_V(x,w^*,j)$.
    \end{enumerate}
    By witness indistinguishability, for any $(x,w) \in \calR$ we have that $\trdist{\View_{V}(x,w,j) - S(x,j)} \leq \wieps$. This concludes the first direction.

    Let $(P,V)$ be a $m$-message QIP with a simulator $S(x,j)$ such that $\trdist{\View_{V}(x,w,j) - S(x,j)} \leq \wieps/2$ for any $j \in [k]$ and $(x,w) \in \calR$.
    Fix an $x$ in the language and let $w_0,w_1 \in \calR(x)$.
    By the assumption, we have 
    $\trdist{\View_{V}(x,w_b,j) - S(x,j)} \leq \eps/2$ for all $b \in \zo$.
    By the triangle inequality, we have
    \[ \trdist{\View_{V}(x,w_0,j) - \View_{V}(x,w_1,j)} \leq \trdist{\View_{V}(x,w_0,j) - S(x,j)} + \trdist{\View_{V}(x,w_0,j) - S(x,j)} \leq \wieps. \]
    This concludes the second direction and the proof.
\end{proof}

\begin{fact}[WI and ZK with unbounded simulation]\label{fact:mvWIZK}
    A QIP $(P,V)$ is
    an $\QSWI\big(m,1-\ceps, \seps, \wieps\big)$ protocol if and only if for any efficient verifier $V^*$ there exists a (potentially unbounded) simulator $S$ such that for any $(x,w) \in \R$ and any $j \in [k]$:
    \[ \trdist{\View_{V^*}(x,w,j) - S(x,j)} \leq O(\wieps). \]
\end{fact}
The proof of \cref{fact:mvWIZK} is completely analogous to that of \cref{fact:WIZK}

\section{Fundamental Results on QSWI}
The main goal of this section is to prove that $\hvQSWI = \QSWI$.
The strategy is simple:
\cref{fact:WIZK,fact:mvWIZK} are fundamental facts about witness indistinguishability that also apply in the classical setting, not just the quantum setting.
However, they are especially useful for QSWI due to Kobayashi's \cite{Kob08} direct transformations for QSZK proofs.
In particular, we observe that Kobayashi's transformations preserve the prover's efficiency, and they all work with unbounded zero-knowledge simulators (e.g., neither the prover nor the verifier have to run the simulator).
This means that those transformations can be applied to our ``efficient prover'' and ``unbounded simulator'' setting, and thus they give analogous results for QSWI.
The only thing we have to work harder for is keeping track of the witness indistinguishability error (i.e., the simulator error).
In total, we prove the following:
\begin{enumerate}
    \item Any honest-verifier QSWI proof can be compiled into an honest-verifier QSWI proof with $3$ messages ($2$ rounds).
    \item Any honest-verifier QSWI proof with $3$ messages can be compiled into an honest-verifier \emph{public-coin} QSWI proof with $3$ messages, in which the verifier sends only $1$ classical (uniform) bit.
    \item Any honest-verifier public-coin QSWI proof with $3$ messages in which the verifier sends only $1$ classical (uniform) bit, is also malicious-verifier QSWI.
\end{enumerate}
The remainder of this section is dedicated to proving those transformations while also keeping track of the WI error.
It turns out the WI error in each transformation blows up by at most a polynomial factor.
This will allow us to use the transformation in \cref{sec:batch-to-qswi}.

\subsection{Round Compression for $\hvQSWI$}

\begin{lemma}\label{lem:round-compr}
    For any polynomial $m\deq m(n)$ and any $\ceps,\seps,\wieps \in [0,1]$ such that $m \geq 4$ and $\ceps < \frac{(1-\seps)^2}{16(m+1)^2}$.
    Then 
    \[\hvQSWI\left(m,1-\ceps,\seps,\wieps\right) \subseteq \hvQSWI\left(3,1-\frac{\ceps}{2},1-\frac{(1-\seps)^2}{32(m+1)^2},m\wieps\right).\]
\end{lemma}
\begin{proof}
    The proof is nearly identical to \cite[Lemma 17]{Kob08} aside from the analysis of witness indistinguishability. We reiterate the protocol and argument for completeness to provide context for our analysis of WI error.

    Let $\calL = \{\calL_\text{yes}, \calL_\text{no}\}$ be an $\NP$ language in $\hvQSWI\left(m,1-\ceps,\seps,\wieps\right)$ with honest verifier $V$ and honest prover $P$. Without loss of generality, assume $m$ is even, since we can modify any protocol with odd $m$ to have an even number of messages by having the verifier send a ``dummy" first message. 
    Let quantum registers $\reg{V}$, $\reg{M}$, and $\reg{P}$ be the verifier's private register, the message channel register, and the prover's private register acting on Hilbert spaces $\calV$, $\calM$, and $\calP$ respectively.
    For $1 \le j \le \frac{m}{2}+1$, let $V_j$ be the $j$th transformation applied by the honest verifier $V$ to registers $(\reg{V}, \reg{M})$. 
    For $1 \le j \le \frac{m}{2}$ let $P_j$ be the $j$th transformation applied by the honest prover $P$ to registers $(\reg{M}, \reg{P})$. 
    Let $\Pi = \{\Pi_{\text{acc}}, \Pi_{\text{rej}}\}$ be the final measurement the verifier performs to decide acceptance and rejection. 
    For $1 \le j \le \frac{m}{2}+1$, let $\ket{\psi_j}$ be the quantum state in $(\reg{V}, \reg{M}, \reg{P})$ just before $V$ applies $V_j$ (so $\ket{\psi_1} = \ket{0_{\calV \otimes \calM \otimes \calP}}$).

    We construct a three-message $\hvQSWI$ protocol for $\calL$ in \cref{fig:hv m to hv 3} using the same approach that Watrous and Kitaev used to show that all quantum interactive proofs only require three messages \cite{KW00}. Informally, our new honest prover $P'$ will send our new honest verifier $V'$ every intermediate state that the old verifier $V$ would have seen when interacting with the old verifier $P$. Then, using only one round of interaction, $V'$ will either check that the final state is accepting or that a random transition is valid. 

    \newprotocol{Honest Verifier's Three-Message Protocol}
    {an instance $x$.}
    {
    \item Receive registers $(\reg{V}_j, \reg{M}_j)$ for $2 \le j \le \frac{m}{2} + 1$ from the prover.
    
    \item Prepare $\ket{0}$ in registers $(\reg{V}_1, \reg{M}_1)$ and $\ket{\Phi^+} = \frac{1}{\sqrt{2}}(\ket{0}\ket{0} + \ket{1}\ket{1})$ in single qubit registers $(\reg{X}, \reg{Y})$. 
    
    Choose $r \in \{1, \dots, \frac{m}{2}\}$ uniformly at random and apply $V_r$ to the qubits in $(\reg{V}_r, \reg{M}_r)$.
    
    Send $r$, $\reg{M}_r$, and $\reg{Y}$ to the prover.

    \item Receive registers $\reg{M}_r$ and $\reg{Y}$ from the prover. Choose $b \in \{0,1\}$ uniformly at random.

    \begin{tabular}{@{}lp{5in}@{}}
        If $b = 0$: & Controlled on $\reg{X}$, perform a controlled-swap between $(\reg{V}_r, \reg{M}_r)$ and $(\reg{V}_{r+1}, \reg{M}_{r+1})$ and a controlled-not on $\reg{Y}$. Apply Hadamard to $X$. Accept if $\reg{X}$ contains $\ket{0}$ and reject otherwise. \\
        If $b = 1$: & Apply $V_{\frac{m}{2} + 1}$ to $(\reg{V}_{\frac{m}{2} + 1}, \reg{M}_{\frac{m}{2} + 1})$. Accept if $(\reg{V}_{\frac{m}{2} + 1}, \reg{M}_{\frac{m}{2} + 1})$ contains an accepting state of the original protocol, and reject otherwise.
    \end{tabular}
    }
    {Honest verifier's three-message protocol}
    {fig:hv m to hv 3}

    To show completeness, suppose $x \in \calL_\text{yes}$ and the honest prover $P'$ is provided some witness $w$ such that $(x,w) \in \R$, where $\R$ is some $\NP$ relation for $\calL$. $P'$ can then use that witness $w$ to simulate the $m$-message interaction between $P$ and $V$ to prepare $\ket{\psi_j}$ in $(\reg{V}_j, \reg{M}_j, \reg{P}_j)$ and send all $(\reg{V}_j, \reg{M}_j)$ for $2 \le j \le \frac{m}{2} + 1$ as its first message to $V'$. 
    $P'$ also prepares $\ket{0}$ in $\reg{P_1}$.
    After $P'$ receives $r$, $\reg{M}_r$, and $\reg{Y}$, $P'$ applies $P_r$ to $(\reg{M}_r, \reg{P}_r)$ and performs a controlled-swap between $\reg{P}_r$ and $\reg{P}_{r+1}$ with $\reg{Y}$ as the control. 
    Finally, $P'$ sends $\reg{M}_r$ and $\reg{Y}$ back to $V'$. 
    All these transformations are efficient because the transformations performed by the original verifier and prover in the $m$-message $\hvQSWI$ protocol were efficient.
    
    If $b=0$, the protocol accepts with probability $1$, because $P_r V_r \ket{\psi_r} = \ket{\psi_{r+1}}$ for all $r \in \{1, \dots, \frac{m}{2}\}$, implying that $(\reg{X}, \reg{Y})$ are returned to their original state $\ket{\Psi^+}$ after the two controlled-swaps (one by the prover and one by the verifier).
    If $b=1$, the protocol accepts with probability at least $1 - \e$, since the final state of the original $m$-message protocol was accepted with probability at least $1-\e$.
    Thus, $V'$ accepts all $x \in \calL_\text{yes}$ with probability at least $1 - \frac{\e}{2}$.
    
    The argument from the proof of \cite[Lemma 17]{Kob08} for soundness being $1-\frac{(1-\seps)^2}{32(m+1)^2}$ carries through without modification.

    Finally, all that is left is witness indistinguishability. 
    By \cref{fact:WIZK}, we just need to construct an unbounded simulator $S'$ of the view of the verifier $V'$ using the fact that there exists an unbounded simulator $S$ of the view of $V$. We can construct each output of $S'$ as follows:
    \begin{itemize}
        \item Let $S'(x,0) = \ket{0_{\calV \otimes \calM}}\bra{0_{\calV \otimes \calM}}$.

        \item Let $S'(x,1)$ be the contents of $(\reg{V}_2, \reg{M}_2, \dots, \reg{V}_{\frac{m}{2} + 1}, \reg{M}_{\frac{m}{2} + 1})$ when we use $S$ to prepare $S(x,j-1)$ in $(\reg{V}_j, \reg{M}_j)$ for each $2 \le j \le \frac{m}{2} + 1$.

        \item Let $S'(x,2)$ be the contents of $(\reg{R}, \reg{X}, \reg{Y}, \reg{V}_2, \reg{M}_2, \dots, \reg{V}_{\frac{m}{2} + 1}, \reg{M}_{\frac{m}{2} + 1})$ when we prepare a uniformly random $r \in \{1, \dots, \frac{m}{2}\}$ in register $\reg{R}$, prepare $\ket{\Phi^+}$ in $(\reg{X}, \reg{Y})$, and then use $S$ to prepare $S(x,r)$ in $(\reg{V}_r, \reg{M}_r)$ and $S(x,j-1)$ in $(\reg{V}_j, \reg{M}_j)$ for each integer $j \in [2, \frac{m}{2} + 1] \setminus \{r\}$. 
    \end{itemize}
    Let $w$ be an arbitrary witness for $x$. Clearly, $\trdist{\View_{V'}(x,w,0) - S'(x,0)} = 0$. By the subadditivity of trace distance for tensor products (for all $\rho_1,\rho_2,\sigma_1, \sigma_2$,  we have $\trdist{\rho_1 \otimes \sigma_1 - \rho_2 \otimes \sigma_2} \le \trdist{\rho_1 - \rho_2} + \trdist{\sigma_1 - \sigma_2}$), and the fact that $S$ has a WI error of $\wieps$, we get that
    \[ \trdist{\View_{V'}(x,w,1) - S'(x,1)} \leq \sum_{j=2}^{\frac{m}{2} + 1} \trdist{\View_{V}(x,w,j) - S(x,j)} = O(m\wieps). \]
    Since $\reg{R}$, $\reg{X}$, and $\reg{Y}$ are prepared in $S'(x,2)$ to exactly match the states they take on in $\View_{V'}(x,w,2)$, then similar to before, we get the following by subadditivity of trace distance
    \[ \trdist{\View_{V'}(x,w,2) - S'(x,2)} \leq O(m\wieps). \]
\end{proof}

\begin{lemma}\label{lem:par-rep}
    For any $\ceps,\seps,\wieps \in [0,1]$ such that $1-\ceps > \seps$ and any polynomial $c\deq c(n) \geq 1$.
    Then 
    \[\hvQSWI\left(3,1-\ceps,\seps,\wieps\right) \subseteq \hvQSWI\left(3,\Big(1-\frac{\ceps}{2}\Big)^c,\seps^c,c\cdot\wieps\right).\]
\end{lemma}
\begin{proof}
    Fix an $\NP$ relation $\calR$ and its corresponding language $\calL$.
    Let $(P,V)$ be a $\hvQSWI\left(3,1-\ceps,\seps,\wieps\right)$ protocol for $\calR$.
    Consider the protocol $(P',V')$ which runs $c$ independent copies of $(P,V)$, and $V'$ accepts if and only if all copies accept.
    Fix any $(x,w) \in \calR$. The probability that the new protocol accepts when run on $(x,w)$ is at least $(1-{\ceps}/{2})^c$ by the independence of the $c$ copies.
    Fix any $x \notin \calL$. By standard QIP parallel repetition techniques (e.g., see \cite[Theorem 6]{KW00}), the probability that the verifier accepts is at most $\seps^c$.

\noindent
    We move on to proving that $(P',V')$ has WI error $k \wieps$.
    Note that by definition, for all $x,w,j$
    \[ \View_{V'}(x,w,j) = \View_{V}(x,w,j)^{\otimes c}. \]
    By \cref{fact:WIZK}, there exists an unbounded simulator $S$ for $(P,V)$. Consider the following simulator $S'$ for $(P',V')$: on input $x,j$, simply runs $S(x,j)$ a total of $c$ independent times and outputs whatever they output.
    That is, the distribution $S'(x,j)$ is simply $S(x,j)^{\otimes c}$.
    We prove that
    $\trdist{\View_{V'}(x,w,j) - S'(x,j)} \leq c \wieps$ by a standard hybrid argument.
    Fix $(x,w) \in \calR$ and $j \in [k]$.
    Consider the following sequence of distributions:
    \[ H_i = S(x,j)^{\otimes i} \circ \View_{V}(x,w,j)^{\otimes c-i}. \]
    Clearly, $H_0 = \View_{V'}(x,w,j)$ and $H_k = S'(x,j)$.
    For any $i \in [k]$,
    by 
    the fact that $H_i$ and $H_{i-1}$ differ only in coordinate $i$ and the fact that trace distance is preserved under tensoring with fixed states, we have
    \[\trdist{H_i - H_{i-1}} = \trdist{\View_{V}(x,w,j) - S(x,j)}.\]
    By the guarantee on $S(x,j)$, we have
    \[\trdist{\View_{V}(x,w,j) - S(x,j)} \leq \wieps.\]
    By the triangle inequality $\trdist{H_k - H_{0}} \leq c\cdot \wieps$, and the claim follows.
\end{proof}

\begin{lemma}\label{lem:3round-negl-error}
    There exists a polynomial $\poly(\cdot,\cdot)$ such that for any $\wieps \in [0,1]$ and any polynomials $m\deq m(n)$ and $p\deq p(n)$ such that $\poly(m,p)\cdot\wieps < 1$, we have that 
    \[\hvQSWI\left(m,2/3,1/3,\wieps\right) \subseteq \hvQSWI\left(3,1-2^{-\Omega(p)},2^{-\Omega(p)},\poly(m,p)\cdot\wieps\right).\]
    Specifically, we have $\hvQSWI \subseteq \hvQSWI\left(3,1-2^{-p},2^{-p},\negl(n)\right)$.
\end{lemma}
\begin{proof}
    Fix an $\NP$ relation $\calR$ and its corresponding language $\calL$.
    Let $(P,V)$ be a $\hvQSWI\left(m,2/3,1/3,\wieps\right)$ protocol for $\calR$.
    Let $(P',V')$ be the protocol that we get by sequentially repeating $(P,V)$ a total of $p'=O(p^2)$ times, and taking the majority vote.
    This is a $\hvQSWI\left(m\cdot p',1-2^{-p^2},2^{-p^2},p'\cdot\wieps\right)$ protocol for $\calR$, and
    completeness, soundness and WI all follow from a similar argument to that of \cref{lem:par-rep}.
    Denote $M = 32(m\cdot p '+1)^2$. By \cref{lem:round-compr}, there exists a $\hvQSWI\left(3,1-2^{-p^2-1},1-\frac{1-2^{-p^2}}{M},mp'^2\cdot\wieps\right)$ protocol for $\calR$.
    Finally, by applying \cref{lem:par-rep} with $c = p M$ and setting $\poly(m,p) = p M m p'^2$, we get a $\hvQSWI\left(3,1-2^{-\Omega(p)},2^{-\Omega(p)},\poly(m,p)\cdot\wieps\right)$ protocol for $\calR$ and the claim follows.
\end{proof}

\subsection{Private-Coin to Public-Coin for $\hvQSWI$}
\begin{lemma}\label{lem:priv-to-pub}
    For any $\ceps,\seps,\wieps \in [0,1]$
    such that $\seps < (1-\ceps)^2$, and any $\calL \in \hvQSWI(3,1-\ceps,\seps,\wieps)$,
    there exists
    a $\hvQSWI\left(3,1-\frac{\ceps}{2},\frac{1}{2} + \frac{\sqrt{\seps}}{2},\wieps\right)$ protocol for $\calL$, where the honest verifier's message consists of a single uniform classical bit.
\end{lemma}
\begin{proof}
    The proof is nearly identical to \cite[Lemma 20]{Kob08} aside from the analysis of witness indistinguishability. We reiterate the protocol and argument for completeness to provide context for our analysis of WI error.

    Let $\calL = \{\calL_\text{yes}, \calL_\text{no}\}$ be an $\NP$ language in $\hvQSWI\left(3,1-\ceps,\seps,\wieps\right)$ with honest verifier $V$ and honest prover $P$.  
    Let quantum registers $\reg{V}$, $\reg{M}$, and $\reg{P}$ be the verifier's private register, the message channel register, and the prover's private register acting on Hilbert spaces $\calV$, $\calM$, and $\calP$ respectively.
    Let $P_1$ and $P_2$ be the honest prover's first and second transformations applied to $(\reg{M}, \reg{P})$, and let $V_1$ and $V_2$ be the verifier's first and second transformations applied to $(\reg{V}, \reg{M})$ between and after the prover's two transformations. After applying $V_2$, the verifier performs some measurement to decide whether to accept or reject.

    We construct a three-message public coin $\hvQSWI$ protocol for $\calL$ in \cref{fig:hv 3 to hv 3 pub} using the same approach that Marriott and Watrous used to compile three-message private-coin quantum interactive proofs into three-message public-coin quantum interactive proofs \cite{MW05}. Informally, our new honest prover $P'$ will send our new honest verifier $V'$ the state that the old verifier $V$ would have in its private register $\reg{V}$ while waiting for the old honest prover $P$ to apply $P_2$. Then, $V'$ will flip a coin to ask $P'$ for the state of $\reg{M}$ either before or after $P$ applied $P_2$. $V'$ would then be able to either check that the state before $P_2$ correctly rewinds to the initial state $\ket{0}$ or that the state after $P_2$ correctly transforms into an accepting final state.

    \newprotocol{Honest Verifier's Three-Message Public-Coin Protocol}
    {an instance $x$.}
    {
    \item Receive a register $\reg{V}$ from the prover.
    
    \item Choose $b \in \{0,1\}$ uniformly at random. Send $b$ to the prover. 
    
    \item Receive register $\reg{M}$ from the prover.

    \begin{tabular}{@{}lp{5in}@{}}
        If $b = 0$: & Apply $V_1^{-1}$ to $(\reg{V}, \reg{M})$. Accept if $\reg{V}$ contains $\ket{0}$ and reject otherwise. \\
        If $b = 1$: & Apply $V_2$ to $(\reg{V}, \reg{M})$. Accept if $(\reg{V}, \reg{M})$ contains an accepting state of the original protocol, and reject otherwise.
    \end{tabular}
    }
    {Honest verifier's three-message public-coin protocol}
    {fig:hv 3 to hv 3 pub}

    To show completeness, suppose $x \in \calL_\text{yes}$ and the honest prover $P'$ is provided some witness $w$ such that $(x,w) \in \R$, where $\R$ is some $\NP$ relation for $\calL$. $P'$ can then use that witness $w$ to simulate the interaction between $P$ and $V$ to prepare $V_1 P_1 \ket{0_{\calV \otimes \calM \otimes \calP}}$ in registers $(\reg{V}, \reg{M}, \reg{P})$. $P'$ then sends $\reg{V}$ as its first message to $V'$. 
    $P'$ then receives $b$. If $b=1$, then $P'$ applies $P_2$ to $(\reg{M}, \reg{P})$. $P$ then sends $\reg{M}$ to $V'$ as its final message.
    All these transformations are efficient because the transformations performed by the original verifier and prover in the original three-message public-coin $\hvQSWI$ protocol were efficient.
    
    If $b=0$, the protocol accepts with probability $1$, because $\tr_{\calM \otimes \calP} V_1^{-1} V_1 P_1 \ket{0_{\calV \otimes \calM \otimes \calP}} = \ket{0_{\calV}}$.
    If $b=1$, the protocol accepts with probability at least $1 - \e$, since the final state of the original $m$-message protocol was accepted with probability at least $1-\e$.
    Thus, $V'$ accepts all $x \in \calL_\text{yes}$ with probability at least $1 - \frac{\e}{2}$.
    
    The argument from the proof of \cite[Theorem 5.4]{MW05} for soundness being $\frac{1}{2} + \frac{\sqrt{\seps}}{2}$ carries through without modification.

    Finally, all that is left is witness indistinguishability. 
    By \cref{fact:WIZK}, we just need to construct an unbounded simulator $S'$ of the view of the verifier $V'$ using the fact that there exists an unbounded simulator $S$ of the view of $V$. We can construct each output of $S'$ as follows:
    \begin{itemize}
        \item Let $S'(x,0) = \ket{0_{\calV \otimes \calM}}\bra{0_{\calV \otimes \calM}}$.

        \item Let $S'(x,1) = \tr_{\calM} V_1 S(x,1) V_1^{-1}$.

        \item Let $S'(x,2) = \frac{1}{2} \ketbra{0}{0} \otimes S(x,2) + \frac{1}{2} \ketbra{1}{1} \otimes (V_1 S(x,1) V_1^{-1})$.
    \end{itemize}
    Note that for any arbitrary witness $w$ for $x$, 
    \begin{itemize}
        \item $\View_{V'}(x,w,0) = \ket{0_{\calV \otimes \calM}}\bra{0_{\calV \otimes \calM}}$.

        \item $\View_{V'}(x,w,1) = \tr_{\calM} V_1 \View_{V}(x,w,1) V_1^{-1}$.

        \item $\View_{V'}(x,w,2) = \frac{1}{2} \ketbra{0}{0} \otimes \View_{V}(x,w,2) + \frac{1}{2} \ketbra{1}{1} \otimes (V_1 \View_{V}(x,w,1) V_1^{-1})$.
    \end{itemize}
    Since $S$ has WI error $\wieps$, we conclude $S'$ has WI error $\wieps$.
\end{proof}

\subsection{Robustness to Malicious Verifiers}
\begin{lemma}\label{lem:malicious-verifier}
    For any $\hvQSWI(3,1-\ceps,\seps,\wieps)$ protocol where the verifier's message consists of a single classical bit, the protocol is actually malicious-verifier QSWI, i.e., it is a $\QSWI(3,1-\ceps,\seps,\wieps)$ protocol.
\end{lemma}
\begin{proof}
    The proof is nearly identical to \cite[Lemma 28]{Kob08}.

    Let $\calL = \{\calL_\text{yes}, \calL_\text{no}\}$ be an $\NP$ language with a $\hvQSWI\left(3,1-\ceps,\seps,\wieps\right)$ protocol between honest verifier $V$ and honest prover $P$ where the verifier's message consists of a single classical bit. Let $\reg{P}_1$ and $\reg{P}_2$ denote the first and second messages sent by $P$. Let $\reg{B}$ be the register containing the single classical bit that $V$ sends to $P$, and let $\reg{B}'$ be the local copy of $\reg{B}$ that the honest verifier holds on to. 
    
    By \cref{fact:WIZK}, there exists an unbounded simulator $S$ of the view of $V$ during the interaction between $V$ and $P$.
    By \cref{fact:mvWIZK}, we can show witness indistinguishability against malicious verifiers by simply constructing an unbounded simulator $S'$ of the view of an arbitrary malicious verifier $V'$. 
    Note that any unbounded simulator $S$ can be thought of as a unitary operation over a large Hilbert space.

    Without loss of generality, we can assume a malicious verifier's single bit message is classical, since the honest prover can measure register $\reg{B}$ in the standard basis to enforce this. 
    Because of this, we can also assume without loss of generality that the malicious verifier measures $\reg{B}$ in the standard basis before sending it and stores a local copy in a private classical register $\reg{B}'$, since any malicious verifier that doesn't do this can be simulated by one that does. 
    
    An arbitrary malicious verifier $V'$ might have additional auxiliary private space we will denote register $\reg{A}$. This auxiliary space could be initialized with helpful quantum information $\rho$.
    Let $V'_1$ be the malicious verifier's first transformation, on registers $(\reg{B}, \reg{P}_1, \reg{A})$. Note that, as said above, $V'$ will measure $\reg{B}$ and copy the classical outcome to a newly created register $\reg{B}'$ as the final step of the transformation $V'_1$, so we will say $V'_1$ outputs registers $(\reg{B}, \reg{B}', \reg{P}_1, \reg{A})$. 
    
    Our construction for $S'$ is as follows:
    \begin{itemize}
        \item Let $S'(x,0) = \rho$.
        \item Let $S'(x,1) = S(x,1) \otimes \rho$.
        \item Let $S'(x,2)$ be the contents of $(\reg{B}'_{\textsf{mv}}, \reg{P}_1, \reg{P}_2, \reg{A})$ after the following procedure:
        \begin{enumerate}
            \item Store $\rho$ in $\reg{A}$. Initialize all other registers used below to $\ket{0}$.

            \item Use $S$ to prepare $S(x,2)$ in registers $(\reg{B}'_{\textsf{hv}}, \reg{P}_1, \reg{P}_2)$. 
            
            Here, we envision the action of the simulator $S$ on inputs $x$ and $2$ as a unitary $U_{S(x,2)}$ over registers $(\reg{B}'_{\textsf{hv}}, \reg{P}_1, \reg{P}_2, \reg{S})$. We will keep the contents of $\reg{S}$ around so that we can apply $U_{S(x,2)}^{-1}$ in step 6 below. 
            
            Note that $\reg{P}_1$ is untouched by the honest verifier until the second message from the prover is received, so the contents of $\reg{P}_1$ will be the same here as they were in $S(x,1)$.

            \item Apply $V'_1$ to $(\reg{B}, \reg{P}_1, \reg{A})$ to get $(\reg{B}, \reg{B}'_{\textsf{mv}}, \reg{P}_1, \reg{A})$, the view of the malicious verifier just before it sends its message $\reg{B}$ to the prover.
            
            We know that the contents of registers $\reg{B}$ and $\reg{B}'_{\textsf{mv}}$ match, but for the simulated transcript we're building to be consistent, we need the bit in $\reg{B}'_{\textsf{hv}}$ that the honest prover used to produce $\reg{P}_1$ to match the bit $\reg{B}$ that our simulated malicious verifier plans to send to the prover.
            
            \item Write the $\XOR$ of $\reg{B}'_{\textsf{hv}}$ and $\reg{B}$ in a new register $\reg{X}$.
            
            \item Post-select on $\reg{X}$ containing $\ket{0}$. 
            
            Because our simulator can be unbounded, there are many ways to achieve this. 
            
            \item Output the contents of $(\reg{B}'_{\textsf{mv}}, \reg{P}_1, \reg{P}_2, \reg{A})$.
        \end{enumerate}
    \end{itemize}

    Since $\View_{V'}(x,j)$ can be constructed by following the same steps as the simulator but replacing each call to $S$ with a call to $\View_V$, we find that for all $j$,
    \[
    \trdist{\View_{V'}(x,j) - S'(x,j)} \le \wieps,
    \]
    so witness indistinguishability error is preserved.
\end{proof}

Finally, by combining \cref{lem:3round-negl-error,lem:priv-to-pub,lem:malicious-verifier}, we get the main theorem of this section:
\begin{theorem}\label{thm:hvQSWIeqQSWI}
    There exists a polynomial $\poly(\cdot,\cdot)$ such that for any $\wieps \in [0,1]$ and any polynomials $m\deq m(n)$ and $p\deq p(n)$ such that $\poly(m,p)\cdot\wieps < 1$, we have that 
    \[\hvQSWI\left(m,2/3,1/3,\wieps\right) \subseteq \pubQSWI\left(3,1-2^{-\Omega(p)},2^{-\Omega(p)},\poly(m,p)\cdot\wieps\right).\]
    Specifically, we have $\hvQSWI \subseteq \pubQSWI\left(3,1-2^{-p},2^{-p},\negl(n)\right) = \pubQSWI$.
\end{theorem}

\section{Quantum Batch Proofs}\label{sec:batch-to-qswi}
The ultimate goal of this section is to prove that quantum batch proofs can be compiled into QSWI proofs.
\subsection{Definitions and Preliminaries}
We start by defining batch relations for $\NP$.
\begin{definition}
    For any $\NP$ relation $\calR$, we define the batch relation $\batch{\calR}{t}$ as follows:
    \[ \batch{\calR}{t} = \set{ \Big((x_1,\ldots,x_t), (w_1,\ldots,w_t)\Big) : \forall i \in [t], (x_i,w_i) \in \calR }. \]
\end{definition}
A quantum batch proof is a QIP for the relation $\batch{\calR}{t}$.
The compression parameter of the batch proof is roughly the total communication compared to the trivial protocol where the prover just sends all of the $\NP$ witnesses to the verifier. Formally:
\begin{definition}[Quantum Batch Proofs]
    Let $\calR$ be some $\NP$ relation and $t(n)$ be some polynomial.
    A QIP $(P,V)$ for the batch relation $\batch{\calR}{t}$ is called a quantum batch proof.
    Denote by $k\deq k(n,t)$ the number of rounds, and by $q_{M} \deq q_{M}(n,t)$ the number of message qubits.
    For any $\rho \in \N$ (which could be a function of $n,t$),
    we say that $(P,V)$ is $\rho$ compressing if $q_{M}(nt) \cdot k(nt) = o(\rho t)$.
\end{definition}

Let $(P,V)$ be a $k$-round QIP for some batch relation $\batch{\R}{t}$.
Note that for any $j \in [k]$, $\View_V(\cdot, \cdot, j)$ is a mapping of the instances and witness to the view of the verifier in round $j$.
For any $(x_1,\ldots,x_t) \in \batch{\R}{t}$ and each $i \in [t]$, let $w_{i}^{0}, w_{i}^{0}$ be two valid witness for $x_i$.
Denote $\xv = (x_1,\ldots,x_t)$, $\wv^0 = (w^0_1,\ldots,w^0_t)$, and $\wv^1 = (w^1_1,\ldots,w^1_t)$.
We define the mapping $f_j \deq f_{\xv,\wv^0,\wv^1,j}$ that inputs $t$ classical bits and outputs a view:
\[ f_j(b_1,\ldots,b_t) = \View_V\big((x_1,\ldots,x_t),(w^{b_1}_1,\ldots,w^{b_t}_t), j \big). \]

\paragraph{Game Theory.}
Consider a two-player zero-sum game $G = (R,C,p)$ where $R$ is the set of pure strategies for the first (``row'') player, $C$ is the set pure strategies for the second (``column'') player, and $p : R \times C \to [0,1]$ is the payoff function for the row player.
That is, the row player wants to maximize $p$, whereas the column players wants to minimize $p$.
Let $\rho$ and $\kappa$ denote mixed strategies for the row and column players respectively. That is, $\rho$ is a distribution over $R$ and $\kappa$ is a distribution over $C$.
We denote by $v(G)$ the value of the game, which is defined as
\[ v(G) = \max_{\sigma}\min_{\tau} \E_{
\begin{array}{cc}
     r \sim \sigma \\
     c \sim \tau 
\end{array}}
\big[p(r,c)\big] = \min_{\tau}\max_{\sigma} \E_{
\begin{array}{cc}
     r \sim \sigma \\
     c \sim \tau 
\end{array}}
\big[p(r,c)\big].
\]
\noindent
We use the following ``sparse'' minimax theorem due to Lipton and Yung \cite{LY94}:
\begin{lemma}[Sparse Minimax \cite{LY94}]\label{lem:sparseMM}
    For any $\eps > 0$, there exists a multiset $S \subseteq C$ of size $\Theta(\log \abs{R}/\eps^2)$ such that for all $r \in R$:
    \[ \E_{c \getsr S} \big[p(r,c)] \leq v(G) + \eps. \]
\end{lemma}

\subsection{Quantum Batching Implies QSWI}

\begin{theorem}\label{thm:batch-to-wi}
    Let $\calR$ be an $\NP$ relation and let $\rho < 1$.
    If $\calR$ has a $\rho$-compressing QBP with $m$ messages, completeness error $\ceps$, and soundness error $\seps$ then the corresponding language $\calL$
    is in $\hvQSWI\big(m,1-\ceps, \seps, O(\sqrt{\rho})\big)$, with a non-uniform honest prover.
\end{theorem}

\begin{proof}
Let $(P,V)$ be the batch proof, and let $t=t(n)$ be the batch size.
Let $\calD = \set{\calD_n}_{n\in \N}$ be a distribution ensemble where $\calD_n$ is a distribution over the support:
\[ \set{(x,w_0,w_1) : \abs{x} = n \;\wedge\; (x,w_0) \in \calR \;\wedge\; (x,w_1) \in \calR } \]
The QSWI proof $(P', V')$ works as follows. 

\newprotocol{QSWI Prover Algorithm}
{an instance $x$ and a witness $w$.}
{
    \item Samples $i^* \getsr [t]$.
    \item For each $i \in [t] \setminus \set{i^*}$, samples $(x_i,w^0_i,w^i,_1) \sim \calD_n$ and a uniform bit $b_i \getsr \zo$.
    \item Sets $(x_{i^*}, w_{i^*}) = (x,w)$ and $w_i = w_{i}^{b_i}$ for all $i \in [t] \setminus \set{i^*}$. Denote $\xv = (x_1,\ldots,x_t), \wv = (w_1,\ldots,w_t)$.
    \item Sends $\xv$ and $j$ to the verifier.
    \item Afterwards, each circuit $P'_j(x,w)$ is simply defined as the circuit $P_j(\xv,\wv)$.
}{Algorithm $P'$}{fig:prover-qswi-from-batch}

\newprotocol{QSWI Verifier Algorithm}
{an instance $x$.}
{
    \item Receives $\xv$ and $j$ from the prover.
    \item Each circuit $V'_j(x,w)$ is simply defined as the circuit $V_j(\xv,\wv)$.
}{Algorithm $V'$}{fig:verifier-qswi-from-batch}

Completeness and soundness follow immediately from the completeness and soundness of $(P,V)$ and the fact that the prover samples ``yes'' instances with valid witnesses.

We move on to honest-verifier witness indistinguishability.
The goal is to find a distribution $\calD_n$ for which WI holds.
For any $j \in [k]$, we first define the following two-player zero-sum game $G=(R,C,p)$:
\begin{itemize}
    \item The row player chooses $(x,w^0,w^1)$ such that $\abs{x} = n$ and $(x,w^0),(x,w^1) \in \calR$.
    That is, the set of pure strategies is
    \[R = \set{(x,w^0,w^1) : \abs{x} = n \;\wedge\; (x,w^0) \in \calR \;\wedge\; (x,w^1) \in \calR}.\]
    
    \item The column player chooses $t$ tuples $(x_i,w^0_i,w^1_i)$ such that $\abs{x_i} = n$ and $(x_i,w^0_i),(x_i,w^1_i) \in \calR$ for all $i \in [t]$.
    That is, the set of pure strategies is
    \[C = \set{\big((x_i,w^0_i,w^1_i)\big)_{i\in[t]} : \forall i\in[t],\;\abs{x} = n \;\wedge\; (x_i,w^0_i)\in \calR \;\wedge\;(x_i,w^1_i) \in \calR}.\]

    \item Fix some $r=(x,w^0,w^1) \in R$ and $ \big((x_i,w^0_i,w^1_i)\big)_{i\in[t]} \in C$, and
    denote $\xv = (x_1\ldots,x_t)$
    , $\wv^0 = (w^0_1,\ldots,w^0_t)$ and $\wv^1 = (w^1_1,\ldots,w^1_t)$.
    For any $i^* \in [t]$, denote by $\xv|_{i^*\gets x}$
    the vector that we get by replacing coordinate $i^*$ in $\xv$ with $x$.
    Next, define the function $f_{r,c,i^*}(b) = \View_V(\xv|_{i^*\gets x}, \wv(b)|_{i^*\gets w^{b_{i^*}}}, j)$
    where $b = b_1\cdots b_t \in \zo^t$
    and $\wv(b) = (w_1^{b_1},\ldots,w_t^{b_t})$.
    The payoff of the game is defined as:
    \[ p\Big(r=(x,w^0,w^1), c=\big((x_i,w^0_i,w^1_i)\big)_{i\in[t]}\Big) = \E_{{i^*} \getsr [t]} \Big[ \trdist{f_{r,c,i^*}(U^t|_{i^* \gets 0}) - f_{r,c,i^*}(U^t|_{i^* \gets 1})} \Big].  \]
\end{itemize}

\begin{proposition}\label{prop:game-value}
    $v(G) = O(\sqrt{\rho})$.
\end{proposition}
\begin{proof}
    Let $\sigma$ be a distribution over $R$.
    We define the distribution $\tau \deq \sigma^{\otimes t}$, i.e., taking $t$ independent samples from $\sigma$.
    In this case, the expected payoff is
    \[\E_{\sigma,\tau} \left[\E_{{i^*} \getsr [t]} \Big[ \trdist{f_{\sigma,\tau,i^*}(U^t|_{i^* \gets 0}) - f_{\sigma,\tau,i^*}(U^t|_{i^* \gets 1})} \Big]\right].\]
    We would like to immediately use \cref{lem:compress}, but the compressing function $f_{\sigma,\tau,i^*}$ depends on the coordinate $i^*$ which we are fixing. Therefore, we have to remove this dependence first.

    \noindent
    By linearity of expectation, the previous expression is equal to
    \[ \E_{{i^*} \getsr [t]}\left[\E_{\sigma,\tau} \Big[ \trdist{f_{\sigma,\tau,i^*}(U^t|_{i^* \gets 0}) - f_{\sigma,\tau,i^*}(U^t|_{i^* \gets 1})} \Big]\right].\]
    By definition of $\sigma,\tau$, we have that $x$ and $x_{i^*}$ are identically distributed, therefore
    $\xv|_{i^* \gets x}$ is distributed identically to $\xv$.
    Similarly,
    we have that $w^\beta$ and $w^\beta_{i^*}$ are identically distributed for all $\beta \in \zo$,
    therefore
    $\wv(b)|_{i^* \gets w^{b_{i^*}}}$ is distributed identically to $\wv(b)$.
    Therefore
    we can replace $f_{\sigma,\tau,i^*}$ with the following function $f_{\tau}(b) = \View_V(\xv, \wv(b), j)$ (that does not replace coordinate $i^*$ with the sample $\sigma$) without changing the distribution nor the expected value.
    That is, the above is equal to
    \[ \E_{{i^*} \getsr [t]}\left[\E_{\tau} \Big[ \trdist{f_{\tau}(U^t|_{i^* \gets 0}) - f_{\tau}(U^t|_{i^* \gets 1})} \Big]\right].\]
    Again by linearity of expectation:
    \[ \E_{\tau}\left[ \E_{{i^*} \getsr [t]}\Big[ \trdist{f_{\tau}(U^t|_{i^* \gets 0}) - f_{\tau}(U^t|_{i^* \gets 1})} \Big]\right].\]
    Now by \cref{lem:compress}, we get that the above is upper bounded by $O(\sqrt{\rho})$ and \cref{prop:game-value} follows.
\end{proof}
Now by \cref{lem:sparseMM,prop:game-value}, for $\eps = \sqrt{\rho} > 0$ there exists a subset $S_n \subseteq C$ of size $O(\log\abs{R}/\rho) = \poly(n)$ such that for any $(x,w^0,w^1) \in R$,
\[ \E_{\big((x_i,w^0_i,w^1_i)\big) \getsr S} \left[ p\Big((x,w^0,w^1), \big((x_i,w^0_i,w^1_i)\big)_{i\in[t]}\Big)\right] \leq v(G) + \eps = O(\sqrt{\rho}). \]

This means that the non-uniform prover that samples the additional instances and witnesses from $S_n$ can guarantee value $O(\sqrt{\rho})$.
But this value is exactly the witness indistinguishability error and \cref{thm:batch-to-wi} follows.
\end{proof}

As a corollary of \cref{thm:hvQSWIeqQSWI,thm:batch-to-wi}, we get that any language that has a sufficiently compressing batch proof, has a \emph{malicious-verifier} QSWI proof:
\begin{corollary}
    There exists a polynomial $\poly(\cdot)$ and a negligible function $\negl(n)$ such that for any $\rho \in (0,1)$ and any polynomial $m\deq m(n)$ such that $\poly(m)\cdot\sqrt{\rho} < 1$,
    if $\calL \in \NP$ has a $\rho$-compressing QBP with $m$ messages then $\calL \in \pubQSWI\left(3,1-\negl(n),\negl(n),\poly(m)\cdot\sqrt{\rho}\right)$, with a non-uniform honest prover.
\end{corollary}

\section{Open Problems} 
\begin{itemize}
    \item Does every $\NP$ problem have QSWI proofs? \cref{thm:batch-to-wi} shows that constructing quantum batch proofs for all of $\NP$ would imply $\NP$ has QSWI proofs with inverse polynomial error. 
    
    See \cref{appendix:npbatchattempt} for our attempt to design quantum batch proofs for $\NP$.
    Unfortunately, our protocol is insecure against a cheating prover keeping a private register entangled with the communication qubits.
    
    \item Can every QSWI proof be made to have perfect completeness? Solving the following toy problem would imply yes:

    \begin{problem}
        Construct an efficient quantum circuit that uses polynomially many copies of $\sqrt{p}\ket{0} + \sqrt{1-p}\ket{1}$ to exactly produce the state $\sqrt{1-\frac{c}{p}} \ket{0} + \sqrt{\frac{c}{p}} \ket{1}$ for some known efficiently computable constant $c$ and unknown $p$.
    \end{problem}
\end{itemize}

\section*{Acknowledgments}
The authors would like to thank Scott Aaronson for teaching the graduate course that inspired this work and for guidance and encouragement throughout this project.
We would also like to thank Jeff Champion, Joshua Cook and William Kretschmer
for helpful conversations.
Shafik Nassar is supported by NSF awards CNS-2140975 and CNS-2318701.
Ronak Ramachandran is supported by an NSF Graduate Research Fellowship. This work was done in part while Ronak Ramachandran was interning at Sandia National Laboratories.

\bibliographystyle{alpha}
\bibliography{main}

\appendix
\section{Failed Attempt at Quantum Batch Proofs for $\NP$}
\label{appendix:npbatchattempt}

By \cref{thm:batch-to-wi}, if every problem in $\NP$ has $\rho$-compressing quantum batch proofs, then $\NP$ is in $\QSWI$ with $O(\sqrt{\rho})$ WI error. 
More concretely, if an efficient quantum verifier possessing $k$ instances $(x_1,\dots, x_k)$ of any arbitrary $\NP$ problem can verify that all $k$ instances are yes-instances through a $\rho k$-qubit interaction with an efficient quantum prover (who potentially possesses $k$ valid witnesses $(w_1,\dots, w_k)$), then every $\NP$ problem would have a QSWI protocol with witness indistinguishability error $O(\sqrt{\rho})$. 

What follows is an attempt to design quantum batch proofs for $\NP$ by having the verifier search for instances with no valid witness using a distributed Grover's search algorithm. Unfortunately, the protocol is not secure against dishonest provers that keep a private register entangled with the messages sent to the verifier.
If the protocol had worked, it would have been $\frac{1}{\sqrt{k}}$-compressing, resulting in $k^{-1/4}$ WI error.
We note that the protocol does not put any restriction on $k$, so we can reduce the WI error by making $k$ larger (and thus paying in the total communication).

\subsection{The Distributed Grover Search Protocol}

Consider an arbitrary language $\calL$ in $\NP$. Such a language must have a polynomial-time classical verifier $A$ such that for any yes-instance $x$, $A$ accepts $(x,w)$ for some witness $w$, whereas for any no-instance $x$, $A$ rejects $(x,w)$ for all witnesses $w$. Let $\text{REJ}(x,w)$ be $1$ if $A$ rejects $(x,w)$ and $0$ otherwise.

As usual, both the prover and verifier have access to $k$ instances $(x_1,\dots, x_k)$ of $\calL$ and the honest prover additionally has valid witnesses for each instance, $(w_1,\dots, w_k)$. This protocol will use $3$ quantum registers: the communicated index, the witness, and the verifier index.
The first two registers are public communication registers that the verifier and prover take turns acting on, while the verifier index register is kept private by the verifier.
We use an explicit $\otimes$ to distinguish the public communication registers from the verifier's private register.
For an honest prover, the protocol works as follows:

\figbox{Distributed Grover Search Subroutine}
{
First, the verifier chooses $b \in \{0,1\}$ and $j \in [k]$ uniformly at random.

The verifier initializes $\ket{\phi_0} = \sum_{i \in [k]} \ket{i}\ket{0} \otimes \alpha_{0,i} \ket{i}$ where $\alpha_{0,i} = \frac{1}{\sqrt{k}}$.

For $t = 0,\ldots,T$:

\begin{enumerate}
    \item The prover transforms the state to $\sum_{i \in [k]} \ket{i}\ket{w_i} \otimes \alpha_{t,i}\ket{i}$.
    
    \item To ensure the index registers match, the verifier performs CNOT from the verifier index to the communicated index, 
    rejects if the communicated index is not $\ket{0}$, 
    then performs the same CNOT again.
    
    The verifier transforms the state to
    $\sum_{i \in [k]} \ket{i}\ket{w_i} \otimes (-1)^{f_i}\alpha_{t,i}\ket{i}$, 
    where $f_i = \text{REJ}(x_i,w_i)$ if $b = 0$, and $f_i = \delta_{j,i}$ otherwise.
    
    \item The prover uncomputes the witness register,
    leaving the verifier with the state $\sum_{i \in [k]} \ket{i}\ket{0} \otimes (-1)^{f_i}\alpha_{t,i}\ket{i}$.
    
    \item The verifier rejects if anything other than $\ket{0}$ is observed in the witness register.
    
    The verifier performs CNOT from the verifier index to the communicated index, 
    rejects if the communicated index is not $\ket{0}$,
    performs the Grover diffusion operator on the verifier index,
    and finally performs the same CNOT again
    to get a new state 
    $\ket{\phi_{t+1}} = \sum_{i \in [k]} \ket{i}\ket{0} \otimes \alpha_{t+1,i} \ket{i}$.
\end{enumerate}

Finally, the verifier measures the verifier index register.
Let $i^*$ be the result.

If $b=0$, then the verifier asks the prover for $w_{i^*}$ and accepts if $A$ accepts $(x_{i^*}, w_{i^*})$, rejecting otherwise.

If $b=1$, the verifier accepts if $i^* = j$, rejecting otherwise.
}{The main subroutine of our failed attempt to quantumly batch all of $\NP$}{fig:batching-np}

One can easily verify every step above can be accomplished by an efficient quantum circuit. 

Step 2 can be seen as the verifier implementing a phase oracle, marking all invalid instances in the $b=0$ case and marking nothing but $j$ in the $b=1$ case. 

One run of the above protocol with $T=\floor{\frac{\pi}{4}\sqrt{k}}$ iterations would suffice if we knew there was only a single marked element. Whenever we don't know the number of marked elements, we can do $O(\log k)$ runs, exponentially increasing the number of iterations as $T = 0,1,2,4,8, \dots, \floor{\frac{\pi}{4}\sqrt{k}}$, until a marked element is found (perhaps randomly interleaving iterations to prevent a dishonest prover from knowing which run any given iteration is from). This still uses only a $O(\frac{1}{\sqrt{k}})$ fraction of the communication that would have been required for the prover to transmit all $k$ witnesses.

Note that in the case of an honest prover, we know the number of marked elements is $0$ in the $b=0$ case and $1$ in the $b=1$ case. We could modify our Grover's search subroutine to give our protocol perfect completeness using techniques from \cite{hoyer00} that make Grover's search succeed at finding marked elements with certainty when the number of marked elements is known.

\subsection{An Attack}

The protocol described above is successful against multiple dishonest prover strategies. Ensuring the index registers match in steps 2 and 4 forces the prover to provide some witness for every instance, ensuring all no-instances will be marked by the verifier in step 2 when $b=0$. Any prover strategy that would result in a negligible probability of the verifier finding an invalid witness in the $b=0$ case (for instance, by marking or unmarking indices at step 3) would result in at least a constant probability of the verifier not finding $j$ in the $b=1$ case.

However, we were unable to prevent the following attack:
Assume that there is just one bad instance at index $i^*$. The prover first initializes a private quantum register to $\ket{0}$. This private prover register must be big enough to store the total number of rounds of interaction with the verifier.
At every iteration of step $3$, controlled on the communicated index register containing $\ket{i^*}$, the prover writes the number of rounds of interaction that have elapsed so far into this new private register.
In this way, every time some amplitude accumulates on $\ket{i^*}$ within the branch of the quantum state where the prover's new private register contains $\ket{0}$, we isolate that amplitude within a brand new branch of the superposition, preventing all future interference. There will be at most $O(\sqrt{k})$ new branches, each with at most $\frac{4}{k}$ probability of being observed, resulting in only a $O(\frac{1}{\sqrt{k}})$ chance of the verifier finding $i^*$ in the $b=0$ case. At the same time, this modification only reduces the probability of observing $j$ in the $b=1$ case by $O(\frac{1}{\sqrt{k}})$.
This results in soundness error that is non-negligibly far from $1$, which is usually good enough since we can amplify soundness by repetition.
However, because we care about the total communication, we cannot afford to amplify in this case.
Specifically, to reduce the soundness error to, say, $1/2$, we would need $O(\sqrt{k})$ repetitions.
This loses all succinctness and brings us back to the trivial $O(k)$ communication, which is not compressing at all.

\end{document}